\documentclass[journal]{IEEEtran}

\usepackage{cite}
\usepackage{pslatex}
\usepackage{moreverb}
\usepackage{epsfig}
\usepackage{color}
\usepackage{amsmath}
\usepackage{amssymb}
\usepackage{booktabs}
\usepackage{array}
\usepackage{multirow}
\usepackage{threeparttable}
\usepackage{url}
\usepackage{dsfont}
\usepackage{epsfig,rotating,setspace,latexsym,amsmath,epsf,amssymb,bm,theorem}
\newtheorem{Theo}{Theorem}
\newtheorem{Lem}{Lemma}
\newtheorem{Def}{Definition}

\DeclareMathAlphabet{\mathpzc}{OT1}{pzc}{m}{it}

\def\calW{{\mathcal W}}

\def\ba{{\mathbf a}}
\def\bb{{\mathbf b}}

\def\bg{{\mathbf g}}
\def\bh{{\mathbf h}}

\def\bx{{\mathbf x}}

\def\bA{{\mathbf A}}

\def\bD{{\mathbf D}}

\def\bI{{\mathbf I}}

\def\bU{{\mathbf U}}

\def\diag{{\rm diag}}

\def\E{{\mathds E}}

\def\b0{{\mathbf 0}}

\begin{document}

\title{On Ergodic Secrecy Capacity of Multiple Input Wiretap Channel with Statistical CSIT}

\author{
\authorblockN{Shih-Chun Lin and  Pin-Hsun Lin}
\thanks{
S.-C. Lin is with Graduate Institute of Computer and Communication
Engineering, National Taipei University of Technology, Taipei,
Taiwan 10643, e-mail: sclin@ntut.edu.tw, P.-H. Lin is with the
Department of Electrical Engineering and Graduate Institute of
Communication Engineering, National Taiwan University, Taipei,
Taiwan, 10617, e-mail: pinhsunlin@gmail.com}
\thanks{This work was supported by the National
Science Council, Taiwan, R.O.C., under Grants
100-2628-E-007-025-MY3.} }

\maketitle 

\begin{abstract}
We consider the secure transmission in ergodic fast-Rayleigh
fading multiple-input single-output single-antenna-eavesdropper
(MISOSE) wiretap channels. We assume that the statistics of both
the legitimate and eavesdropper channels is the only available
channel state information at the transmitter (CSIT). By
introducing a new secrecy capacity upper bound, we prove that the
secrecy capacity is achieved by Gaussian input without prefixing.
To attain this, we form another MISOSE channel for upper-bounding,
and tighten the bound by finding the worst correlations between
the legitimate and eavesdropper channel coefficients. The
resulting upper bound is tighter than the others in the literature
which are based on modifying the correlation between the noises at
the legitimate receiver and eavesdropper. Next, we fully
characterize the ergodic secrecy capacity by showing that the
optimal channel input covariance matrix is a scaled identity
matrix, with the transmit power allocated uniformly among the
antennas. The key to solve such a complicated stochastic
optimization problem is by exploiting the completely monotone
property of the ergodic secrecy capacity to use the stochastic
ordering theory. Finally, our simulation results show that for the
considered channel setting, the secrecy capacity is bounded in
both the high signal-to-noise ratio and large number of transmit
antenna regimes.
\end{abstract}

\vspace{-6mm}

\section{Introduction} \label{Sec_Intro}
The secrecy capacity of a wiretap channel is the maximum
achievable secrecy rate between the transmitter and a legitimate
receiver, and a perfect secrecy constraint is imposed to make no
information be attainable by an eavesdropper
\cite{csiszar1978broadcast}\cite{Wyner_wiretap}. In the wireless
environments, the time-varying characteristic of fading channels
can also be exploited to enhance the secrecy
capacity\cite{gopala2008secrecy}. Further enhancements are
attainable by employing multiple antennas at each node, e.g., in
\cite{Oggier_MIMOME}\cite{Liu_Note_on_secrecy}. However, these
secrecy capacity results
\cite{gopala2008secrecy,Oggier_MIMOME,Liu_Note_on_secrecy} rely on
perfect knowledge of the legitimate receiver's channel state
information at the transmitter (CSIT). Because of the limited
feedback bandwidth and the delay caused by channel estimation, it
may be hard to track the channel coefficients if they vary
rapidly. Thus for fast-fading channels, it is more practical to
consider the case with only partial CSIT of the legitimate
channel. However, in this case, only some lower and upper bounds
of the secrecy capacity are known
\cite{Petropulu_MISO_ergodic}\cite{Rezki_wrong}, and the secrecy
capacity is unknown. Although the general secrecy capacity formula
is reported in \cite{csiszar1978broadcast}, the optimal auxiliary
random variable for prefixing in this formula is \textit{still
unknown}.

In this letter, we consider one important scenario of partial
CSIT, i.e., the transmitter only knows the statistics of both the
legitimate and eavesdropper channels but not the realizations of
them. Under this scenario, we derive the secrecy capacity of the
fast-fading, multiple-input single-output
single-antenna-eavesdropper (MISOSE) wiretap channels, where the
transmitter has multiple antennas while the legitimate receiver
and eavesdropper each have single antenna. Both the coefficients
of the legitimate and eavesdropper channels are Rayleigh faded. We
first propose a new secrecy capacity upper bound to show that the
transmission scheme in \cite{Petropulu_MISO_ergodic} is
secrecy-capacity achieving, which is based on
\cite{csiszar1978broadcast} with Gaussian input but without
prefixing. Then we find the optimal channel input covariance
matrix analytically to fully characterize the ergodic secrecy
capacity, while such a optimization problem is solved numerically
in \cite{Petropulu_MISO_ergodic} without guaranteeing the
optimality. The key is to exploit the completely monotone property
of the ergodic secrecy capacity, then invoking the stochastic
ordering theory \cite{shaked_stochastic_order}.

To obtain a tighter secrecy capacity upper bound than that
reported in \cite{Rezki_wrong}, we introduce another MISOSE
channel with a relaxed secrecy constraint for upper-bounding,
while finding the worst correlations between the coefficients of
the legitimate and eavesdropper channels to tighten the bound. In
\cite{Rezki_wrong}, the upper bound is obtained by directly
applying the concepts from
\cite{gopala2008secrecy}\cite{Oggier_MIMOME} where the correlation
is only introduced between the noises at the legitimate receiver
and eavesdropper and the secrecy constraint is left unchanged.
Note that the secrecy capacity lower bound in \cite{Rezki_wrong}
is indeed not achievable. In order to achieve such a bound, the
variable-rate coding in \cite{gopala2008secrecy} must be invoked,
where the full CSIT of the legitimate channel must be used to vary
the transmission rate in every channel fading state. This can not
be done with only statistical CSIT of the legitimate channel as in
our setting. In addition to the CSIT assumptions, the secrecy
capacity result of \cite{gopala2008secrecy} is builded on the
ergodic slow fading channel assumption where coding among lots of
slow fading channel blocks (each block with lots of coded symbols)
is used. This assumption may be unrealistic owing to the long
latency. For fast fading channels with full CSIT of legitimate
channel and statistical CSIT of the eavesdropper channel, only
some achievability results are known \cite{Li_fading_secrecy_j}.
In contrast to our results, in \cite{Li_fading_secrecy_j}, the
prefixing in \cite{csiszar1978broadcast} may be useful to increase
the secrecy rate. More detailed comparisons between our results
and those in
\cite{Rezki_wrong}\cite{Petropulu_MISO_ergodic}\cite{gopala2008secrecy}
can be found in Remarks 1 and 2.

\section{System Model}\label{Sec_system_model}
In the considered MISOSE wiretap channel, we study the problem of
reliably communicating a secret message $w$ from the transmitter
to the legitimate receiver subject to a constraint on the
information attainable by the eavesdropper (in upcoming
\eqref{eq_equivocation_given_h}). The received signals $y$ and $z$
at the legitimate receiver and eavesdropper (each with single
antenna) from the transmitter equipped with multiple-antenna, can
be represented respectively as \footnote{ In this letter,
$\|\ba\|$ is the vector norm of vector $\mathbf{a}$. The trace and
complex conjugate transpose of matrix $\mathbf{A}$ is denoted by
$\mbox{Tr}(\mathbf{A})$ and $\bA^{\mathrm{H}}$, respectively. The
diagonal matrix is denoted by diag(.). The zero-mean complex
Gaussian random vector with covariance matrix $\Sigma$ is denoted
as $CN(0, \Sigma)$. For random variables (vectors) $A$ and $B$,
$p(A)$ is the probability distribution function (PDF) of $A$,
$I(A;B)$ denotes the mutual information between them while $H(A|B)
$ denotes the conditional differential entropy. We use $A
\rightarrow B \rightarrow C$ to represent that $A,B$, and $C$ form
a Markov chain. All the logarithm operations are of base 2 such
that the unit of rates is in bit.}
\begin{align}
y&=\mathbf{h}^H\mathbf{x}+n_y, \label{eq_legitimate channel}\\
z&=\mathbf{g}^H\mathbf{x}+n_z,\label{eq_eve channel}
\end{align}
 where
$\mathbf{x}$ is a $N_t \times 1$ complex vector representing the
transmitted vector signal, $N_t$ is the number of transmit
antennas, while $n_y$ and $n_z$ are independent and identically
distributed (i.i.d.) circularly symmetric additive white Gaussian
noise with zero mean and unit variance at the legitimate receiver
and eavesdropper, respectively. In \eqref{eq_legitimate channel}
and \eqref{eq_eve channel}, $\mathbf{h}$ and $\mathbf{g}$ are both
$N_t \times 1$ complex vectors, and representing the channels from
the transmitter to the legitimate receiver and eavesdropper,
respectively.

In this work, the channels are assumed to be fast Rayleigh fading.
That is, $\mathbf{h}\sim CN(0, \sigma_{\bh}^2 \mathbf{I})$ and
$\mathbf{g}\sim CN(0, \sigma_{\bg}^2 \mathbf{I})$, respectively,
while $\mathbf{h}$, $\mathbf{g}$, $n_y$ and $n_z$ are independent.
The channel coefficients change in every symbol time. We assume
that the legitimate receiver knows the instantaneous channel state
information of $\mathbf{h}$ perfectly, while the eavesdropper
knows those of $\mathbf{h}$ and $\mathbf{g}$ perfectly. As for the
CSIT, only the distributions of $\mathbf{h}$ and $\mathbf{g}$ are
known at the transmitter, while the realizations of $\mathbf{h}$
and $\mathbf{g}$ are unknown. Thus the transmitter is subjected to
a power constraint as
\begin{equation} \label{power_cons}
\mathrm{Tr}(\Sigma_{\bx}) \leq P,
\end{equation}
where $\Sigma_{\bx}$ is the covariance matrix of $\mathbf{x}$ in
\eqref{eq_legitimate channel} and \eqref{eq_eve channel}.

The perfect secrecy and secrecy capacity are defined as follows.
Consider a $(2^{NR}, N)$-code with an encoder that maps the
message $w\in \calW_N=\{1,2,\ldots, 2^{NR}\}$ into a length-$N$
codeword; and a decoder at the legitimate receiver that maps the
received sequence $y^N$ (the collections of $y$ over the code
length $N$) from the legitimate channel \eqref{eq_legitimate
channel} to an estimated message $\hat w\in\calW_N$. We then have
the following definitions. \vspace{-2mm}
\begin{Def}[Secrecy Capacity
\cite{gopala2008secrecy}] \label{Def_Perfect} {\it Perfect secrecy
is achievable with rate $R$ if, for any $\varepsilon'>0$, there
exists a sequence of $(2^{NR}, N)$-codes and an integer $N_0$ such
that for any $N>N_0$
\begin{equation}\label{eq_equivocation_given_h}
H(w|z^N,\mathbf{h}^N, \mathbf{g}^N)/N> R-\varepsilon', \mathrm{and
}\;\;{\rm Pr}(\hat{w}\neq w) \leq \varepsilon',
\end{equation}
where $w$ is the secret message, $z^N$,$\mathbf{h}^N$, and
$\mathbf{g}^N$ are the collections of $z, \mathbf{h}$, and
$\mathbf{g}$ over code length $N$, respectively. The {\bf secrecy
capacity} $C_s$ is the supremum of all achievable secrecy rates.}

\end{Def}

\vspace{-6mm}

\section{ Secrecy capacity of the MISOSE fast Rayleigh fading wiretap channel}\label{Sec_new_UB}

In this section, we fully characterize the secrecy capacity of the
MISOSE fast Rayleigh fading channel in the upcoming Theorem
\ref{Theorem_No_CQI}. Before that, we present the following
Theorem \ref{Lem_No_CQI} which shows that the scheme in
\cite{Petropulu_MISO_ergodic}, which uses Gaussian $\bx$ without
prefixing in \cite{csiszar1978broadcast}, is capacity achieving.
By introducing new bounding techniques, we obtain tighter secrecy
capacity upper bound than that in \cite{Rezki_wrong} to attain the
secrecy cpapcity. For such a upper bound, we form a better
degraded MISOSE channel of \eqref{eq_legitimate
channel}\eqref{eq_eve channel} with a less stringent perfect
secrecy constraint than \eqref{eq_equivocation_given_h} (in the
upcoming \eqref{eq_equivocation_no_h}), and tighten the upper
bound by carefully introducing correlations to the channels $\bh$
and $\bg$ (in the upcoming \eqref{eq_same_dist_channel_nor2}).

\begin{Theo} \label{Lem_No_CQI}
For the MISOSE fast Rayleigh fading wiretap channel
\eqref{eq_legitimate channel}\eqref{eq_eve channel} with the
statistical CSIT of $\mathbf{h}$ and $\mathbf{g}$, using Gaussain
$\bx$ without prefixing is the optimal transmission strategy. And
the non-zero secrecy capacity $C_s$ is obtained only when
$\sigma_{\bh}>\sigma_{\bg}$, which is
\begin{equation}
C_s=\!\max_{\Sigma_{\bx}}\!\left(\E_{\bh}\!\left[\!\log\left(1+\bh^H\Sigma_{\bx}\bh\!\right)\!\right]\!-\!\E_{\bg}\!\left[\!\log\!\left(1+\bg^H\Sigma_{\bx}\bg\!\right)\!\right]\!\right)\!,
\label{eq_MISOSE_capa_beam}
\end{equation}
where $\Sigma_{\bx}$ is the covariance matrix of the Gaussian
channel input $\bx$ and subject to \eqref{power_cons}, while
$\mathbf{h}\sim CN(0, \sigma_{\bh}^2 \mathbf{I})$ and
$\mathbf{g}\sim CN(0, \sigma_{\bg}^2 \mathbf{I})$.
\end{Theo}

\begin{proof}
From \cite{Petropulu_MISO_ergodic}, we know that the
right-hand-side (RHS) of \eqref{eq_MISOSE_capa_beam} is achievable
and serves as a secrecy capacity lower-bound. Now we present our
new secrecy capacity upper bound which matches the RHS of
\eqref{eq_MISOSE_capa_beam}. The key for establishing such an
upper-bound is first forming a better MISOSE channel than
\eqref{eq_legitimate channel}\eqref{eq_eve channel} in terms of
higher secrecy capacity, and applying the results in
\cite{csiszar1978broadcast}. First, we consider a better scenario
where the eavesdropper does not know the realizations of $\bh$,
and \eqref{eq_equivocation_given_h} becomes
\begin{equation} \label{eq_equivocation_no_h}
H(w|z^N,\mathbf{g}^N)/N> R-\varepsilon', \mathrm{and }\;\;{\rm
Pr}(\hat{w}\neq w) \leq \varepsilon'.
\end{equation}
As in \cite{Caire_fading}, equivalently, we can respectively treat
the output of the legitimate channel as $(y,\bh)$ while that of
the eavesdropper channel as (z,$\bg$). From
\cite{Liang_same_marginal}, the secrecy capacity under constraint
\eqref{eq_equivocation_no_h} is only related to the marginal
distributions $p(y,\bh|\mathbf{x})$ and $p(z,\bg|\mathbf{x})$.
Then for any joint conditional PDF $p_{y',\bh',z,\bg|\mathbf{x}}$
such that $p(y',\bh'|\mathbf{x})=p(y,\bh|\mathbf{x})$, the secrecy
capacity under constraint \eqref{eq_equivocation_no_h} is the
same. Now we introduce our same marginal legitimate channel
$p(y',\bh'|\mathbf{x})$ for \eqref{eq_legitimate channel}, which
is formed by replacing $\bh$ in y with
$\bh'=(\sigma_{\bh}/\sigma_{\bg})\mathbf{g}$ as
\begin{equation} \label{eq_same_dist_channel_nor}
y'=(\bh')^H\mathbf{x}+n_y=((\sigma_{\bh}/\sigma_{\bg})\mathbf{g})^H\mathbf{x}+n_y.
\end{equation}
Since $\bx$ is independent of $\mathbf{h}$ and $\mathbf{g}$ due to
our CSIT assumption, both $\bh'$ and $\mathbf{h}$ have the same
conditional distributions condition on $\bx$ (which equal to
$CN(0, \sigma_{\bh}^2 \mathbf{I})$). Then we know that
$p(y',\bh'|\mathbf{x})$=$p(y,\bh|\mathbf{x})$. From
\eqref{eq_same_dist_channel_nor}, we will focus on the MISOSE
channel in the following,
\begin{equation} \label{eq_same_dist_channel_nor2}
 y''=\mathbf{g}^H\mathbf{x}+(\sigma_{\bg}/\sigma_{\bh})n_y, \;\;\;
z=\mathbf{g}^H\mathbf{x}+n_z.
\end{equation}

We now use the secrecy capacity of the MISOSE channel
\eqref{eq_same_dist_channel_nor2}, under constraint
\eqref{eq_equivocation_no_h}, to upper-bound that of the original
channel \eqref{eq_legitimate channel}\eqref{eq_eve channel}.
Again, we can treat as $(y'',\bg)$ and $(z,\bg)$ as the output of
the legitimate channel and that of the eavesdropper channel in
\eqref{eq_same_dist_channel_nor2}, respectively. As in
\cite{Li_fading_secrecy_j}, we apply this fact into the secrecy
capacity formula of \cite{csiszar1978broadcast},
\begin{align}
C_s &\leq
\max_{p_{(U,\bx)}}I(U;y'',\bg)-I(U;z,\bg)\label{eq_ub_sm0}
\\ &=\max_{p_{(U,\bx)}}I(U;y''|\bg)-I(U;z|\bg), \label{eq_ub_sm}
\\ & \leq \max_{p_{\mathbf{x}}}I(\bx;y''|z,\bg). \label{eq_ub_sm2}
\end{align}
where $U$ in \eqref{eq_ub_sm0} is an auxiliary random variable for
prefixing, which forms the Markov chain $U \rightarrow \bx
\rightarrow (y'',z,\bg)$; the equality \eqref{eq_ub_sm} follows
from \cite{Caire_fading} by the independence of $U$ and $\bg$ due
to our CSIT assumptions; and the inequality \eqref{eq_ub_sm2} is
from \cite{Liu_Note_on_secrecy}. When $\sigma_{\bg}<\sigma_{\bh}$,
the equivalent channel \eqref{eq_same_dist_channel_nor2} is
degraded, and $\bx \rightarrow y'' \rightarrow z$ given $\bg$.
From \cite{Cover_IT}, apply the Markov chain property to
\eqref{eq_ub_sm2}
\begin{equation} \label{eq_ub_sm3}
C_s \leq \max_{p_{\mathbf{x}}}I(\bx;y''|\bg)-I(\bx;z|\bg).
\end{equation}
From \cite{Oggier_MIMOME}, we know that Gaussian $\bx$ is optimal
for the upper bound in \eqref{eq_ub_sm3}, and the upper-bound in
\eqref{eq_ub_sm3} matches the RHS of \eqref{eq_MISOSE_capa_beam}
when $\sigma_{\bg}<\sigma_{\bh}$. Note that when
$\sigma_{\bg}<\sigma_{\bh}$, the RHS of
\eqref{eq_MISOSE_capa_beam} is positive. In contrast, when
$\sigma_{\bg}\geq\sigma_{\bh}$, the upper bound in
\eqref{eq_ub_sm2} is zero since from
\eqref{eq_same_dist_channel_nor2}, $\bx \rightarrow z \rightarrow
y'' $ given $\bg$. And it concludes the proof.
\end{proof}

\textit{Remark 1}: When the transmitter additionally knows the
realizations of $\mathbf{h}$, e. g. \cite{gopala2008secrecy}, the
legitimate channel \eqref{eq_same_dist_channel_nor} is not a same
marginal channel of \eqref{eq_legitimate channel}. In this case,
given $\bx$, $\mathbf{h}$ may not be Gaussian but
$\mathbf{h}'=(\sigma_{\bh}/\sigma_{\bg})\mathbf{g}$ is Gaussian,
and $p(y',\bh'|\mathbf{x})$ from \eqref{eq_same_dist_channel_nor}
may not equal to $p(y,\mathbf{h}|\mathbf{x})$ from
\eqref{eq_legitimate channel}, In our case, the CSIT assumption
makes $\bx$ independent of $\bh$ and $\bg$, then the legitimate
channel in \eqref{eq_same_dist_channel_nor} has the same marginal
as that in \eqref{eq_legitimate channel}. Then we can get rid of
the unrealistic ergodic \textit{slow fading} assumptions in
\cite{gopala2008secrecy} and find the secrecy capacity in
\textit{fast fading} channel. Note that the upper-bound in
\cite{Rezki_wrong} is obtained by directly applying the
derivations in \cite{gopala2008secrecy} to fast fading channel,
and is looser than the upper-bound \eqref{eq_ub_sm2} which is
based on the channel \eqref{eq_same_dist_channel_nor2} and is
tightened by the ``worst'' correlation between $\bh$ and $\bg$
($\bh=\bg$).
\\

Now we show that the optimal $\Sigma_{\bx}$ of
\eqref{eq_MISOSE_capa_beam} is $\diag\{P/N_T,\dots,P/N_T\}$, and
fully characterize the secrecy capacity as follows.

\begin{Theo} \label{Theorem_No_CQI}
For the MISOSE fast Rayleigh fading wiretap channel
\eqref{eq_legitimate channel}\eqref{eq_eve channel} with the
statistical CSIT of $\mathbf{h}$ and $\mathbf{g}$, under power
constraint $P$, the non-zero secrecy capacity $C_s$ is obtained
only when $\sigma_{\bh}>\sigma_{\bg}$, which is
\begin{equation} \label{eq_MISOSE_capa}
C_s=\E_{\bh}\! \left[ \! \log
\!\left(1+P\frac{||\mathbf{h}||^2}{N_T}\!\right)\!\right]\!-\!\E_{\bg}\!\left[\!\log\!\left(1+P\frac{||\mathbf{g}||^2}{N_T}\!\right)\!\right],
\end{equation}
where $\mathbf{h}\sim CN(0, \sigma_{\bh}^2 \mathbf{I})$,
$\mathbf{g}\sim CN(0, \sigma_{\bg}^2 \mathbf{I})$, and $N_T$ is
the number of transmit antennas.
\end{Theo}
\begin{proof}
Subjecting to \eqref{power_cons}, after substituting
$\mathbf{h}\sim CN(0, \sigma_{\bh}^2 \mathbf{I})$ and
$\mathbf{g}\sim CN(0, \sigma_{\bg}^2 \mathbf{I})$ into the
optimization problem in \eqref{eq_MISOSE_capa_beam}, it becomes
\begin{equation} \label{eq_Cs_optimizae}
\max_{\Sigma_{\bx}}\!\!\left(\!\!\!\E_{\mathbf{g}}\!\left[\log\frac{\sigma_{\mathbf{g}}^2/\sigma_{\mathbf{h}}^2+\mathbf{g}^H\Sigma_{\mathbf{x}}\mathbf{g}}{\sigma_{\mathbf{g}}^2/\sigma_{\mathbf{h}}^2}\right]\!-\!\E_{\mathbf{g}}\!\left[\log(1+\mathbf{g}^H\Sigma_{\mathbf{x}}\mathbf{g})\right]\!\!\!\right).
\end{equation}
By using the eigenvalue decomposition
$\Sigma_{\bx}=\bU\mathbf{D}\bU^H$, where $\bU$ is unitary and
$\mathbf{D}$ is diagonal, finding the optimal $\Sigma^*_{\bx}$ of
\eqref{eq_Cs_optimizae} is equivalent to solving
\begin{align}
&\max_{\bU,\,\mathbf{D}}\!\!\left(\!\E_{\mathbf{g}}\!\!\left[\!\log\left(\sigma_{\bg}^2/\sigma_{\bh}^2+\mathbf{g}^H\bU\mathbf{D}\bU^H\mathbf{g}\right)\right]\!\!-\!\!\E_{\mathbf{g}}\!\!\left[\log(1+\mathbf{g}^H\bU\mathbf{D}\bU^H\mathbf{g})\!\right]\!\right),
\notag \\
=&\max_{\mathbf{D}}\left(\E_{\mathbf{g}}\!\!\left[\log\left(\sigma_{\bg}^2/\sigma_{\bh}^2+\mathbf{g}^H\mathbf{D}\mathbf{g}\right)\right]-\E_{\mathbf{g}}\!\!\left[\log(1+\mathbf{g}^H\mathbf{D}\mathbf{g})\right]\right),
\label{EQ_Cs_power_allocation}
\end{align}
where the equality comes from the fact that the distribution of
$\mathbf{g}\sim CN(0, \sigma_{\bg}^2 \mathbf{I})$ is unchanged by
the rotation of unitary $\bU$, and we can set $\Sigma_{\bx}=\bD$
($\bU=\bI$) without loss of optimality.

In the following, we show that subjecting to $\mathrm{Tr}(\bD)
\leq P$, the optimal $\bD$ for \eqref{EQ_Cs_power_allocation} is
\begin{equation} \label{eq_optimial_input}
\mathbf{D}^*=\mbox{diag}\{P/N_T,P/N_T,\cdots,P/N_T\}.
\end{equation}
First of all, from \cite[Section V]{Petropulu_MISO_ergodic}, the
optimal $\bD$ for \eqref{EQ_Cs_power_allocation} satisfies
$\mathrm{Tr}(\bD)=P$. For any
$\mathbf{D}=[d_1,d_2,\cdots,d_{N_T}]$ where $\sum{d_i}=P$ and $d_i
\geq 0, \forall i$, we want to prove that for $\bD^*$ defined in
\eqref{eq_optimial_input}
\begin{align}
&\E_{\mathbf{g}}\left[\log\left(a+\mathbf{g}^H\mathbf{D}\mathbf{g}\right)\right]-\E_{\mathbf{g}}\left[\log\left(1+\mathbf{g}^H\mathbf{D}\mathbf{g}\right)\right]\notag\\\leq
&\E_{\mathbf{g}}\left[\log\left(a+\mathbf{g}^H\mathbf{D}^*\mathbf{g}\right)\right]-\E_{\mathbf{g}}\left[\log\left(1+\mathbf{g}^H\mathbf{D}^*\mathbf{g}\right)\right],
\label{EQ_optimal_power_allocation}
\end{align}
where we denote $\sigma_{\bg}^2/\sigma_{\bh}^2$ by $a$, which
belongs to $[0,1)$. Here we introduce some results from the
stochastic ordering theory \cite{shaked_stochastic_order} to
proceed. \vspace{-0.4cm}
\begin{Def}\cite[p.234]{shaked_stochastic_order}\label{Def_completely_mono} A function $\psi(x):[0,\infty)\rightarrow \mathds{R}$ is completely monotone if for all $x>0$ and $n=0,1,2,\cdots,$ its derivative $\psi^{(n)}$ exists and $(-1)^n\psi^{(n)}(x)\geq 0$.
\end{Def}
\vspace{-0.65cm}
\begin{Def}\cite[(5.A.1)]{shaked_stochastic_order} \label{Def_LT}
Let $B_1$ and $B_2$ be two nonnegative random variables such that
$\E[e^{-sB_1}]\geq\E[e^{-sB_2}]$, for all $s>0$. Then $B_1$ is
said to be smaller than $B_2$ in the Laplace transform order,
denoted as $B_1\leq_{LT} B_2$.
\end{Def}
\vspace{-0.65cm}
\begin{Lem}\cite[Th. 5.A.4]{shaked_stochastic_order} \label{Lemma_LT_eq_MG}
Let $B_1$ and $B_2$ be two nonnegative random variables. If
$B_1\leq_{LT}B_2$ then $\E[f(B_1)]\leq\E[f(B_2)]$, where the first
derivative of a differentiable function $f$ on $[0,\infty)$ is
completely monotone, provided that the expectations exist.
\end{Lem}

To prove \eqref{EQ_optimal_power_allocation}, we let
$B_1=\mathbf{g}^H\mathbf{D}\mathbf{g}$,
$B_2=\mathbf{g}^H\mathbf{D}^*\mathbf{g}$, and
$f(x)=\log(a+x)-\log(1+x)$ to invoke Lemma \ref{Lemma_LT_eq_MG}.
It can be easily verified that $\psi$(x), the first derivative of
$f(x)$, satisfies Definition \ref{Def_completely_mono}. More
specifically, the $n$th derivative of $\psi$ meets
\begin{equation}
\psi^{(n)}(x)=\left\{\begin{array}{ll}
\frac{n!}{(a+x)^{n+1}}-\frac{n!}{(1+x)^{n+1}}>0, & \mbox{if $n$ is even,}\\
\frac{-n!}{(a+x)^{n+1}}+\frac{n!}{(1+x)^{n+1}}<0, & \mbox{if $n$ is odd,}\\
\end{array}\right.
\end{equation}
when $x>0$, since $a \in [0,1)$. Now from Lemma
\ref{Lemma_LT_eq_MG} and Definition \ref{Def_LT}, we know that to
prove \eqref{EQ_optimal_power_allocation} is equivalent to proving
$\E[e^{-sB_1}]\geq \E[e^{-sB_2}]$ or $\log
(\E[e^{-sB_1}]/\E[e^{-sB_2}])\geq 0,\,\forall s>0$. From
\cite[p.40]{Mathai}, we know that
\begin{equation}\label{EQ_log_MG}
\log \!\!
\left(\frac{\E[e^{-sB_1}]}{\E[e^{-sB_2}]}\right)\!\!=\!\!\sum_{k=1}^{N_T}\log(1+2d^*_ks)-\sum_{k=1}^{N_T}\log(1+2d_ks).
\end{equation}
To show that the above is nonnegative, we resort to the
majorization theory. Note that
$\sum_{k=1}^{N_T}\log(1+2\check{d}_ks)$ is a Schur-concave
function \cite{Marshall_Inequalities} in
$(\check{d}_1,\ldots,\check{d}_{N_T})$, $\forall s>0$, and by the
definition of majorization \cite{Marshall_Inequalities}
\[
(d_1^*,\,\cdots,d_{N_T}^*)=(P/N_T,\,P/N_T,\,\cdots,\,P/N_T)\prec(d_1,\,d_2,\,\cdots
d_{N_T}),
\]
where $\bb\prec\ba$ means that $\bb$ is majorized by $\ba$. Thus
from \cite{Marshall_Inequalities}, we know that the RHS of
\eqref{EQ_log_MG} is nonnegative, $\forall s>0$. Then
\eqref{EQ_optimal_power_allocation} is valid, and $\bD^*$ is the
optimal $\bD$ for \eqref{EQ_Cs_power_allocation}. Note that
$\mathbf{D}^*$ is also the optimal $\Sigma_{\bx}$ of
\eqref{eq_MISOSE_capa_beam} according to the discussion under
\eqref{EQ_Cs_power_allocation}. Substituting $\bD^*$ in
\eqref{eq_optimial_input} as the optimal $\Sigma_{\bx}$ into the
target function of \eqref{eq_MISOSE_capa_beam}, we have
\eqref{eq_MISOSE_capa}.
\end{proof}

\textit{Remark 2}: In \cite[Sec. VII]{Petropulu_MISO_ergodic}, the
optimal $\Sigma_{\bx}$ for \eqref{eq_MISOSE_capa} is found by an
iterative algorithm without guaranteeing the optimality. The
contribution of Theorem \ref{Theorem_No_CQI} is analytically
finding the optimal $\Sigma_{\bx}$, which equals to $\bD^*$ in
\eqref{eq_optimial_input}, by exploiting the completely monotone
property of the ergodic secrecy capacity and invoking Lemma
\ref{Lemma_LT_eq_MG}. Finally, as discussed in Section
\ref{Sec_Intro}, the secrecy rate lower-bound in
\cite{Rezki_wrong} is not achievable, thus the conclusion in
\cite{Rezki_wrong} that uniform power allocation among transmit
antennas as \eqref{eq_optimial_input} is \textit{not} secrecy
capacity achieving is \textit{wrong}.


\begin{figure}[t]
\centering
\epsfig{file=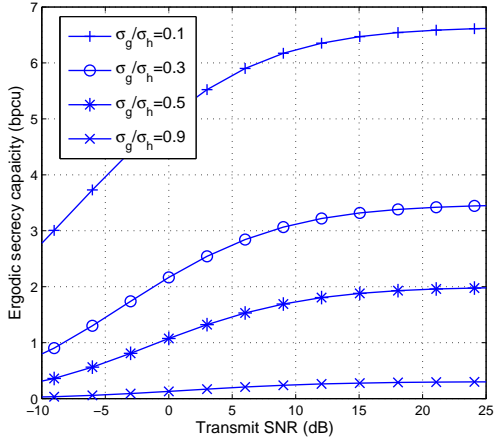,
width=0.42\textwidth} \caption{Comparison of the ergodic secrecy
capacities under different eavesdropper-to-legitimate-channel
qualities $\sigma_{\bg}/\sigma_{\bh}$.}
\label{Fig_rate_different_sigma_ratio}
\end{figure}

\begin{figure}[t]
\centering
\epsfig{file=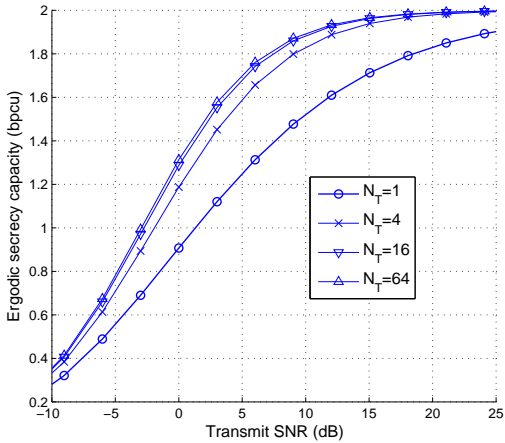, width=0.42\textwidth}
\caption{Comparison of the ergodic secrecy capacities under
different numbers of transmit antennas $N_T$.}
\label{Fig_rate_different_N_T}
\end{figure}
\vspace{-3mm}

\section{Simulation Results}
\vspace{-1mm}

In this section we compare the secrecy capacities under different
channel conditions. The transmit signal-to-noise ratio (SNR) is
defined as $P$ in dB scale since $n_y$ and $n_y$ both have unit
variances. In Fig. \ref{Fig_rate_different_sigma_ratio} we compare
the secrecy capacities with $N_T=2$ under different
$\sigma_{\bg}/\sigma_{\bh}$. The secrecy capacity increases with
decreasing $\sigma_{\bg}/\sigma_{\bh}$. The capacity converges to
$2\log (\sigma_{\bh}/\sigma_{\bg})$ when the SNR is high, which
meets $\eqref{eq_MISOSE_capa}$ with large $P$. In Fig.
\ref{Fig_rate_different_N_T}, we compare the secrecy capacities
with different numbers of transmit antennas $N_T$. We can also
find that the capacity converges when $N_T$ is large enough. This
results can be easily seen by letting $N_T \rightarrow \infty$ in
\eqref{eq_MISOSE_capa}, and applying the central limit theorem on
$||\mathbf{h}||^2/N_T$ and $||\mathbf{g}||^2/N_T$, respectively.

\vspace{-4mm}

\section{Conclusion}\label{Sec_conclusion}
In this paper, we derived the secrecy capacity of the MISOSE
ergodic fast Rayliegh fading wiretap channel, where only the
statistical CSIT of the legitimate and eavesdropper channels is
known. By introducing a new secrecy capacity upper bound, we first
showed that Gaussian input without prefixing is secrecy capacity
achieving. Then we analytically found the optimal channel input
covariance matrix, and fully characterized the secrecy capacity.

\vspace{-4mm}


\bibliographystyle{IEEEtran}
\bibliography{IEEEabrv,SecrecyPs2}

\begin{thebibliography}{10}
\providecommand{\url}[1]{#1}
\csname url@rmstyle\endcsname
\providecommand{\newblock}{\relax}
\providecommand{\bibinfo}[2]{#2}
\providecommand\BIBentrySTDinterwordspacing{\spaceskip=0pt\relax}
\providecommand\BIBentryALTinterwordstretchfactor{4}
\providecommand\BIBentryALTinterwordspacing{\spaceskip=\fontdimen2\font plus
\BIBentryALTinterwordstretchfactor\fontdimen3\font minus
  \fontdimen4\font\relax}
\providecommand\BIBforeignlanguage[2]{{%
\expandafter\ifx\csname l@#1\endcsname\relax
\typeout{** WARNING: IEEEtran.bst: No hyphenation pattern has been}%
\typeout{** loaded for the language `#1'. Using the pattern for}%
\typeout{** the default language instead.}%
\else
\language=\csname l@#1\endcsname
\fi
#2}}

\bibitem{csiszar1978broadcast}
I.~Csisz{\'a}r and J.~K{\.o}rner, ``{Broadcast channels with confidential
  messages},'' \emph{{IEEE} Trans. Inform. Theory}, vol.~24, no.~3, pp.
  339--348, May 1978.

\bibitem{Wyner_wiretap}
A.~D. Wyner, ``The wiretap channel,'' \emph{Bell Syst. Tech. J.}, vol.~54, pp.
  1355--1387, 1975.

\bibitem{gopala2008secrecy}
P.~Gopala, L.~Lai, and H.~El~Gamal, ``{On the secrecy capacity of fading
  channels},'' \emph{{IEEE} Trans. Inform. Theory}, vol.~54, no.~10, pp.
  4687--4698, Oct. 2008.

\bibitem{Oggier_MIMOME}
F.~Oggier and B.~Hassibi, ``The secrecy capacity of the {MIMO} wiretap
  channel,'' \emph{{IEEE} Trans. Inform. Theory}, vol.~57, no.~8, Aug. 2011.

\bibitem{Liu_Note_on_secrecy}
T.~Liu and S.~Shamai, ``A note on the secrecy capacity of the multi-antenna
  wiretap channel,'' \emph{{IEEE} Trans. Inform. Theory}, vol.~55, no.~6, pp.
  2547--2553, Jun. 2009.

\bibitem{Petropulu_MISO_ergodic}
J.~Li and A.~P. Petropulu, ``On ergodic secrecy rate for {G}aussian {MISO}
  wiretap channels,'' \emph{{IEEE} Trans. Wireless Commun.}, vol.~10, no.~4,
  pp. 1176--1187, April 2011.

\bibitem{Rezki_wrong}
Z.~Rezki, F.~Gagnon, and V.~Bhargava, ``The ergodic capacity of the {MIMO}
  wire-tap channel,'' \emph{http://arxiv.org/abs/0902.0189}, 2009.

\bibitem{shaked_stochastic_order}
M.~Shaked and J.~G. Shanthikumar, \emph{Stochastic Orders}.\hskip 1em plus
  0.5em minus 0.4em\relax Springer, 2007.

\bibitem{Li_fading_secrecy_j}
Z.~Li, R.~Yates, and W.~Trappe, ``Achieving secret communication for fast
  {R}ayleigh fading channels,'' \emph{{IEEE} Trans. Wireless Commun.}, vol.~9,
  no.~9, pp. 2792 -- 2799, Sep. 2010.

\bibitem{Caire_fading}
G.~Caire and S.~Shamai, ``On the capacity of some channels with channel state
  information,'' \emph{{IEEE} Trans. Inform. Theory}, vol.~45, no.~6, pp.
  2007--2019, Sept. 1999.

\bibitem{Liang_same_marginal}
Y.~Liang, H.~V. Poor, and S.~S. Shamai, ``Information theoretic security,''
  \emph{Foundations and Trends in Communications and Information Theory},
  vol.~5, no. 4-5, pp. 355--580, Apr. 2009.

\bibitem{Cover_IT}
T.~M. Cover and J.~A. Thomas, \emph{Elements of information theory},
  2nd~ed.\hskip 1em plus 0.5em minus 0.4em\relax Wiley-Interscience, 2006.

\bibitem{Mathai}
A.~M. Mathai and S.~B. Provost, \emph{Quadratic forms in random
  variables}.\hskip 1em plus 0.5em minus 0.4em\relax Marcel Dekker, New York,
  1992.

\bibitem{Marshall_Inequalities}
A.~W. Marshall and I.~Olkin, \emph{Inequalities: theory of majorization and its
  application}.\hskip 1em plus 0.5em minus 0.4em\relax Academic Press, 1980.

\end{thebibliography}

\end{document}